\documentclass{eptcs}
\providecommand{\event}{DICE 2010}

\usepackage[all]{xy}
\usepackage{url}
\usepackage{calc,xspace,enumerate}

\usepackage{amsmath,amssymb,amsbsy,amsfonts,amsthm} 
\usepackage{euscript} 

\usepackage{stmaryrd} 

\theoremstyle{plain}
 \newtheorem{theorem}{Theorem}
 \newtheorem{proposition}{Proposition}
 
 \newtheorem{lemma}{Lemma}

\theoremstyle{definition}
 \newtheorem{definition}{Definition}
 \newtheorem{example}{Example}

\theoremstyle{remark}

\newcommand{\trs}{\mathcal{R}}

\newcommand{\equivalent}{\Leftrightarrow}

\newcommand{\cnot}{\mbox{\textbf{Not}}}
\newcommand{\cvar}{\mbox{\textbf{Var}}}
\newcommand{\Cor}{\mbox{\textbf{Or}}}

\newcommand{\cexists}{\mbox{\textbf{Exists}}}

\newcommand{\fnot}{\mathtt{not}}
\newcommand{\ffor}{\mathtt{or}}
\newcommand{\fand}{\mathtt{and}}
\newcommand{\member}{\mathtt{in}}
\newcommand{\verify}{\mathtt{ver}}

\newcommand{\brouge}{\blacklozenge}

\newcommand{\Ptime}{\mbox{\sc Ptime}}
\newcommand{\NPtime}{\mbox{\sc NPtime}}
\newcommand{\Linspace}{\mbox{\sc linspace}}

\newcommand{\Logspace}{\mbox{\sc Logspace}}

\newcommand{\Pspace}{\mbox{\sc Pspace}}

\newcommand{\functional}{\mbox{\texttt{F}}}

\newcommand{\some}[3]{#1_{#2}, \cdots, #1_{#3}}
\newcommand{\many}[2]{\some{#1}{1}{#2}}

\edef\union{\bigcup}

\newcommand{\size}[1]{|#1|}

\newcommand{\Variables}{\mathcal{X}}

\newcommand{\Functions}{\mathcal{F}}
\def\funone{\texttt{f}}
\def\funtwo{\texttt{g}}

\def\main{\texttt{f}}

\newcommand{\Terms}[1]{\mathbf{T}(#1)}

\newcommand{\Constructors}{\mathcal{C}}
\def\conone{\mathbf{c}}

\newcommand{\Consterms}{\Terms{\Constructors}}
\edef\termone{t}
\edef\termtwo{s}
\edef\termthree{u}
\edef\termfour{v}

\edef\patone{p}

\newcommand{\progone}{\texttt{f}}

\newcommand{\Equations}{\mathcal{E}}

\newcommand{\sem}[1]{\llbracket #1 \rrbracket}

\def\To{\to}
\def\transto{\mbox{${\stackrel{\mbox{\tiny $+\;$}}{\To}}$}}
\def\normto{\mbox{${\stackrel{!\;}{\To}}$}}
\def\reftransto{\mbox{${\stackrel{*\;}{\To}}$}}

\newcommand{\too}{\rightsquigarrow}

\newdimen\proofrulebreadth \proofrulebreadth=.05em
\newdimen\proofdotseparation \proofdotseparation=1.25ex
\newdimen\proofrulebaseline \proofrulebaseline=2ex
\newcount\proofdotnumber \proofdotnumber=3
\let\then\relax
\def\hfi{\hskip0pt plus.0001fil}
\mathchardef\squigto="3A3B
\newif\ifinsideprooftree\insideprooftreefalse
\newif\ifonleftofproofrule\onleftofproofrulefalse
\newif\ifproofdots\proofdotsfalse
\newif\ifdoubleproof\doubleprooffalse
\let\wereinproofbit\relax
\newdimen\shortenproofleft
\newdimen\shortenproofright
\newdimen\proofbelowshift
\newbox\proofabove
\newbox\proofbelow
\newbox\proofrulename
\def\shiftproofbelow{\let\next\relax\afterassignment\setshiftproofbelow\dimen0 }
\def\shiftproofbelowneg{\def\next{\multiply\dimen0 by-1 }%
\afterassignment\setshiftproofbelow\dimen0 }
\def\setshiftproofbelow{\next\proofbelowshift=\dimen0 }
\def\setproofrulebreadth{\proofrulebreadth}

\def\prooftree{
\ifnum  \lastpenalty=1
\then   \unpenalty
\else   \onleftofproofrulefalse
\fi
\ifonleftofproofrule
\else   \ifinsideprooftree
        \then   \hskip.5em plus1fil
        \fi
\fi
\bgroup
\setbox\proofbelow=\hbox{}\setbox\proofrulename=\hbox{}%
\let\justifies\proofover\let\leadsto\proofoverdots\let\Justifies\proofoverdbl
\let\using\proofusing\let\[\prooftree
\ifinsideprooftree\let\]\endprooftree\fi
\proofdotsfalse\doubleprooffalse
\let\thickness\setproofrulebreadth
\let\shiftright\shiftproofbelow \let\shift\shiftproofbelow
\let\shiftleft\shiftproofbelowneg
\let\ifwasinsideprooftree\ifinsideprooftree
\insideprooftreetrue
\setbox\proofabove=\hbox\bgroup$\displaystyle 
\let\wereinproofbit\prooftree
\shortenproofleft=0pt \shortenproofright=0pt \proofbelowshift=0pt
\onleftofproofruletrue\penalty1
}

\def\eproofbit{
\ifx    \wereinproofbit\prooftree
\then   \ifcase \lastpenalty
        \then   \shortenproofright=0pt  
        \or     \unpenalty\hfil         
        \or     \unpenalty\unskip       
        \else   \shortenproofright=0pt  
        \fi
\fi
\global\dimen0=\shortenproofleft
\global\dimen1=\shortenproofright
\global\dimen2=\proofrulebreadth
\global\dimen3=\proofbelowshift
\global\dimen4=\proofdotseparation
\global\count255=\proofdotnumber
$\egroup  
\shortenproofleft=\dimen0
\shortenproofright=\dimen1
\proofrulebreadth=\dimen2
\proofbelowshift=\dimen3
\proofdotseparation=\dimen4
\proofdotnumber=\count255
}

\def\proofover{
\eproofbit 
\setbox\proofbelow=\hbox\bgroup 
\let\wereinproofbit\proofover
$\displaystyle
}%
\def\proofoverdbl{
\eproofbit 
\doubleprooftrue
\setbox\proofbelow=\hbox\bgroup 
\let\wereinproofbit\proofoverdbl
$\displaystyle
}%
\def\proofoverdots{
\eproofbit 
\proofdotstrue
\setbox\proofbelow=\hbox\bgroup 
\let\wereinproofbit\proofoverdots
$\displaystyle
}%
\def\proofusing{
\eproofbit 
\setbox\proofrulename=\hbox\bgroup 
\let\wereinproofbit\proofusing
\kern0.3em$
}

\def\endprooftree{
\eproofbit 
  \dimen5 =0pt
\dimen0=\wd\proofabove \advance\dimen0-\shortenproofleft
\advance\dimen0-\shortenproofright
\dimen1=.5\dimen0 \advance\dimen1-.5\wd\proofbelow
\dimen4=\dimen1
\advance\dimen1\proofbelowshift \advance\dimen4-\proofbelowshift
\ifdim  \dimen1<0pt
\then   \advance\shortenproofleft\dimen1
        \advance\dimen0-\dimen1
        \dimen1=0pt
        \ifdim  \shortenproofleft<0pt
        \then   \setbox\proofabove=\hbox{%
                        \kern-\shortenproofleft\unhbox\proofabove}%
                \shortenproofleft=0pt
        \fi
\fi
\ifdim  \dimen4<0pt
\then   \advance\shortenproofright\dimen4
        \advance\dimen0-\dimen4
        \dimen4=0pt
\fi
\ifdim  \shortenproofright<\wd\proofrulename
\then   \shortenproofright=\wd\proofrulename
\fi
\dimen2=\shortenproofleft \advance\dimen2 by\dimen1
\dimen3=\shortenproofright\advance\dimen3 by\dimen4
\ifproofdots
\then
        \dimen6=\shortenproofleft \advance\dimen6 .5\dimen0
        \setbox1=\vbox to\proofdotseparation{\vss\hbox{$\cdot$}\vss}%
        \setbox0=\hbox{%
                \advance\dimen6-.5\wd1
                \kern\dimen6
                $\vcenter to\proofdotnumber\proofdotseparation
                        {\leaders\box1\vfill}$%
                \unhbox\proofrulename}%
\else   \dimen6=\fontdimen22\the\textfont2 
        \dimen7=\dimen6
        \advance\dimen6by.5\proofrulebreadth
        \advance\dimen7by-.5\proofrulebreadth
        \setbox0=\hbox{%
                \kern\shortenproofleft
                \ifdoubleproof
                \then   \hbox to\dimen0{%
                        $\mathsurround0pt\mathord=\mkern-6mu%
                        \cleaders\hbox{$\mkern-2mu=\mkern-2mu$}\hfill
                        \mkern-6mu\mathord=$}%
                \else   \vrule height\dimen6 depth-\dimen7 width\dimen0
                \fi
                \unhbox\proofrulename}%
        \ht0=\dimen6 \dp0=-\dimen7
\fi
\let\doll\relax
\ifwasinsideprooftree
\then   \let\VBOX\vbox
\else   \ifmmode\else$\let\doll=$\fi
        \let\VBOX\vcenter
\fi
\VBOX   {\baselineskip\proofrulebaseline \lineskip.2ex
        \expandafter\lineskiplimit\ifproofdots0ex\else-0.6ex\fi
        \hbox   spread\dimen5   {\hfi\unhbox\proofabove\hfi}%
        \hbox{\box0}%
        \hbox   {\kern\dimen2 \box\proofbelow}}\doll%
\global\dimen2=\dimen2
\global\dimen3=\dimen3
\egroup 
\ifonleftofproofrule
\then   \shortenproofleft=\dimen2
\fi
\shortenproofright=\dimen3
\onleftofproofrulefalse
\ifinsideprooftree
\then   \hskip.5em plus 1fil \penalty2
\fi
}

\newcommand{\ninfer}[3]
     {\prooftree
          #1 
          \justifies #2
          \using #3
      \endprooftree}

\renewcommand{\phi}{\varphi}
\newcommand{\vide}{\nnill}
\renewcommand{\epsilon}{\varepsilon}

\newcommand{\blanc}{\blacktriangledown}
\newcommand{\noir}{\blacktriangle}

\newcommand{\undef}{\bot}

\newcommand{\Put}{\texttt{put}}
\newcommand{\Union}{\texttt{match}}
\newcommand{\hyp}{\texttt{vhyp}}
\newcommand{\hypo}{\texttt{hyp}}

\newcommand{\listHyp}{\texttt{hypList}}
\newcommand{\append}{\texttt{append}}

\newcommand{\utrue}{\mathbf{T}}
\newcommand{\ufalse}{\mathbf{F}}

\edef\infSt{\prec}
\edef\infEq{\preceq}
\edef\infMulSt{\prec^p}
\edef\infMul{\preceq^p}

\newcommand{\precFS}{\prec_{\Functions}}
\newcommand{\precEqFS}{\preceq_{\Functions}}
\newcommand{\egalFS}{\approx_{\Functions}}

\edef\mpoSt{\prec_{ppo}}

\newcommand{\mpoXSt}{\prec^{p}_{ppo}}

\newcommand{\interp}[1]{\llparenthesis #1\,\rrparenthesis}

\newcommand{\suc}{\mbox{\textbf{s}}}

\newcommand{\zero}{\mbox{\textbf{0}}}

\def\nnill{\mbox{\textbf{nil}}}

\newcommand{\true}{\mbox{\textbf{tt}}}
\newcommand{\false}{\mbox{\textbf{ff}}}

\newcommand{\cons}{\mbox{\textbf{cons}}}

\newcommand{\eval}{\mbox{\texttt{eval}}}

\newcommand{\iif}{\mbox{\texttt{if\ }}}
\newcommand{\tthen}{\mbox{\texttt{then\ }}}
\newcommand{\eelse}{\mbox{\texttt{else\ }}}

\newcommand{\vrai}{\mbox{\tt tt}}
\newcommand{\faux}{\mbox{\tt ff}}

\newcommand{\ee}{\mbox{\tt{e}}}

\newcommand{\MaxPoly}{\textbf{Max-Poly}}

\newcommand{\PPO}{PPO}

\newcommand{\sousterme}{\unlhd}

\title{Observation of implicit complexity by non confluence}
\author{Guillaume Bonfante
\institute{Nancy University}
\email{guillaume(dot)bonfante(at)loria(dot)fr}
}
\def\authorrunning{Guillaume Bonfante}
\def\titlerunning{Observation of implicit complexity by non confluence}

\begin{document}
\maketitle

\footnotetext{Work partially supported by project ANR-08-BLANC-0211-01 (COMPLICE)}

\begin{abstract}
We propose to consider non confluence with respect to implicit complexity. We come back to some well known classes of first-order functional program, for which we have a characterization of their intentional properties, namely the class of cons-free programs, the class of programs with an interpretation, and the class of programs with a quasi-interpretation together with a termination proof by the product path ordering. They all correspond to PTIME. We prove that adding non confluence to the rules leads to respectively PTIME, NPTIME and PSPACE. Our thesis is that the separation of the classes is actually a witness of the intentional properties of the initial classes of programs.
\end{abstract}


In implicit complexity theory, one of the issues is to characterize large classes of programs, not extensionally but intentionally. That is, for a given class of functions, to delineate the largest set of programs computing this class. One of the issues with this problem is that it is hard to compare (classes of) programs. Indeed, strict syntactical equality is definitely too restrictive, but larger (the interesting ones) relations are undecidable. So, comparing theories (defining their own class of programs) is rather complicated. Usually, one gives a remarkable example, illustrating the power of the theory.

We propose here an other way to compare sets of programs. The idea is to add a new feature--in the present settings, non determinism-- to two programming languages. Intuitively, if a function can be computed with this new feature in a language $L_1$ but non in the language $L_2$, we say that $L_1$ is more powerful than $L_2$.  Let us formalize a little bit our intuition.

Suppose for the discussion that programs are written as rewriting systems, that is, programming languages are classes of rewriting systems. Let us say furthermore that a program $p$ is simulated by $q$ whenever each step of rewriting in $t \stackrel{\ell \to r}{\to} u \in p$ can be simulated by a rewriting step $t' \stackrel{\ell' \to r'}{\to} u' \in q$. Equivalence of languages $L_1$ and $L_2$ states that any program in $L_1$ is simulated by a program in $L_2$ and vice versa.

Given a programming language $L$, its non deterministic extension $L.n$ is the programming language obtained by adding to $L$ an oracle $\mathbf{choose}(r_1,r_2)$ which, given two rules $r_1$ and $r_2$ which can be defined in $L$, applies the "right" rule depending on the context. So, given some $p \in L$, both $p \cup \{r_1\}$ and $p \cup \{ r_2\}$ are supposed to be in $L$, but  not necessarily $p \cup \{ r_1, r_2\}$.

Suppose now, that we are given two programming languages $L_1$ and $L_2$ such that functions computed in $L_1.n$ are strictly included in those in $L_2.n$. Then, $L_1$ cannot simulate $L_2$. Ad absurdum, suppose that $L_1$ can simulate $L_2$, take $f$ computed by $p_2 \in L_2.n$ but not in $L_1.n$. Then, each rule $\ell \to r$ in $p_2$ is simulated by a rule $\ell' \to r'$ in some program $p_1 \in L_1$. But then, the derivation $t_1 \to t_2 \to \cdots \to t_n$ in $L_2.n$ can be simulated by some derivation $t'_1Ê\to t'_2 \to \cdots \to t'_n$ in $L_1.n$. This leads to the contradiction.

It is clear that the notion of equivalence we took for the discussion is very strong. However, we believe that the argument would hold in a larger context.

In this paper, we observe three programming languages, 
\begin{itemize}
\item programs with a polynomially bounded constructor preserving interpretation ($\functional.\text{cons}$) which extend cons-free programs, 
\item programs with (polynomial) strict interpretation ($\functional.{\text{SI}}$) and 
\item programs with a quasi-interpretation together with a proof of termination by \PPO, written \functional.{\text{QI}.\PPO}. 
\end{itemize}

These three languages characterize \Ptime. The first one is a new result, the two latter ones are respectively proved in~\cite{BonfanteCichonMarionTouzet01} and \cite{BMM07}. 

\centerline{
\xymatrix{
\functional.{\text{cons}} \ar[d] & \functional.{\text{SI}} \ar[d] & \functional.{\text{QI}.\PPO} \ar[d]\\
\Ptime \ar@{}[r]|{=} & \Ptime \ar@{}[r]|{=} & \Ptime}
}

\vspace{2ex}
Their non deterministic observation characterize  \Ptime, \NPtime \ and \Pspace. The second characterization has been proven in~\cite{BonfanteCichonMarionTouzet01}.

\centerline{
\xymatrix{
\functional.{\text{cons}.n} \ar[d] & \functional.{\text{SI}.n} \ar[d] & \functional.{\text{QI}.\PPO.}n \ar[d]\\
\Ptime \ar@{}[r]|{\stackrel{?}{\neq}} & \NPtime \ar@{}[r]|{\stackrel{?}{\neq}} & \Pspace}
}

\vspace{2ex}
The issue of confluence of Term Rewriting Systems has been largely studied, see for instance 
\cite{terese}.
It benefits from some nice properties, for instance it is modular and algorithms
are given to automatically compute the confluence up to termination. 


It is clear that there is also an intrinsic motivation for a study of non confluent programs. It would not be reasonable to cover all the researches dealing with this issue. But, let us make three remarks. First, since non-confluence can give us some freedom to write programs, it is of interest to observe what new functions  this extra feature allows us to compute. From our result about non-confluent programs in $\functional.{\text{cons}.n}$, one may extract a compilation procedure to compute them "deterministically". Second, Kristiansen and Mender in~\cite{Kristiansen06,Kristiansen09} have proposed a scale --a la Grzegorczyk, using non determinism, to characterize $\Linspace$. Finally, one should keep in mind characterizations in the logical framework. Let us mention for instance the characterization of \Pspace\ given in~\cite{Gaboardi08a}. It is an extension of a characterization of \Ptime, and thus, we think that their construction is a good candidate for observation as presented above. 






\section{Preliminaries}

We suppose that the reader has familiarity with first-order rewriting. We briefly 
recall the context of the theory, essentially to fix the notations. Dershowitz and
Jouannaud's survey~\cite{DJ90} of rewriting is a good entry point 
for beginners. 

Let $\Variables$ denote a (countable) set of \emph{variables}. 
Given a \emph{signature} $\Sigma$, the set of \emph{terms} over $\Sigma$ and
$\Variables$ is denoted by $\Terms{\Sigma,\Variables}$ and the 
set of \emph{ground terms}, that is terms without variables, by $\Terms{\Sigma}$. 

The size $\size{t}$ of a term $t$ is defined as the number of symbols in $t$. 
For example the size of the term $f(a, x)$ is $3$.

A \emph{context} is a term $C$ with a particular variable $\lozenge$. If $t$ is a term, 
$C[t]$ denotes the term $C$ where the variable $\lozenge$ has been replaced
by $t$.

\subsection{Syntax of programs} 

Let $\Constructors$ be a (finite) signature
of \emph{constructor} symbols and $\Functions$ a (finite) signature of\emph{ function symbols}. Thus, we are given an algebra of constructor terms $\Terms{\Constructors, \Variables}$. A rule is a pair $(\ell,r)$, next written $\ell \to r$, where:
\begin{itemize}
\item $\ell = \funone(\many{\patone}{n})$ where $\funone \in \Functions$ and $\patone_i \in \Terms{\Constructors,\Variables}$ for all $i = 1,\ldots, n$,
\item and $r \in \Terms{\Constructors \cup \Functions,\Variables}$ is a term such that any variable occuring in $r$ also occurs in $\ell$. 
\end{itemize}

\begin{definition}
\label{def:untyped-program}
A \emph{program} is a quadruplet $\progone = \langle \Variables, \Constructors, \Functions, \Equations \rangle$ such that $\Equations$ is
a finite set of rules.  
We distinguish among $\Functions$ a main function symbol whose name is
given by the program name $\progone$. $\functional$ denotes the set of programs. 
\end{definition}

The set of rules induces a rewriting relation $\to$. The relation $\transto$ is the transitive closure
of $\to$, and  $\reftransto$ is the reflexive and transitive closure  of $\to$. Finally, we say
that a term $t$ is a \emph{normal form} if there is no term $u$ such that $t \to u$. Given two terms
$t$ and $u$, $t \normto u$ denotes the fact that $t \reftransto u$ and $u$ is a normal form.

All along, when it is not explicitly mentioned, we suppose programs to be \emph{confluent}, 
that is, the rewriting relation is confluent. 

 The domain of the computed functions is 
the constructor term algebra $\Consterms$. The program
$\progone= \langle \Variables, \Constructors, \Functions, \Equations \rangle$ computes a partial function $\sem{\main} : \Consterms^n \to
\Consterms$ defined as follows. For every $\many{\termthree}{n} \in \Consterms,
\sem{\main}(\many{\termthree}{n}) = \termfour$ iff $
\main(\many{\termthree}{n}) \normto \termfour$ and $\termfour$ is a constructor term.  

\begin{definition}[Call-tree]Suppose we are given a program $\langle \Variables, \Constructors, \Functions, \Equations\rangle$. 
Let $\too$ be the relation $$(f,t_1,\ldots, t_n) \too (g,u_1, \ldots, u_m) \Leftrightarrow f(t_1, \ldots, t_n) \to C[g(v_1, \ldots, v_m)] \reftransto C[g(u_1, \ldots, u_m)]$$ where $f$ and $g$ are function symbols, $C$ is a context and $t_1, \ldots, t_n, u_1, \ldots, u_m$ are constructor terms.  Given a function symbol $f$ and constructor terms $t_1, \ldots, t_n$,  the relation $\too$ defines a tree whose root is $(f,t_1, \ldots, t_n)$ and $\eta'$ is a daughter of $\eta$ iff $\eta \too \eta'$. The relation $\too^+$ is the transitive closure of $\too$. 
\end{definition}


\subsection{Interpretations of programs}\label{sec:qi}

Given a signature $\Sigma$, a $\Sigma$-algebra on a domain $A$ is a mapping 
$\interp{-}$ which  associates to every $n$-ary symbol $f \in \Sigma$ an $n$-ary 
function $\interp{f} : A^n \to A$. Such a $\Sigma$-algebra can be extended to terms by:
\begin{itemize}
\item $\interp{x} = 1_{A}$, that is the identity on $A$,  for $x \in \Variables$,
\item $\interp{f(t_1, \ldots, t_m)} = \text{comp}(\interp{f},\interp{t_1},\ldots, \interp{t_m})$ where $\text{comp}$ is the composition of functions.
\end{itemize} 
 Given a term $t$ with $n$ variables, $\interp{t}$  is a function $A^n \to A$.


\begin{definition}\label{def:interpretation}Given an ordered set $(A,<)$ and a program $\langle \Variables, \Constructors, \Functions, \Equations \rangle$, let us consider a  $(\Constructors \cup \Functions)$-algebra $\interp{-}$ on $A$. It  is said to:
\begin{enumerate}
\item be strictly monotonic if for any symbol $f$, the function $\interp{f}$ is a strictly monotonic function, that is if $x_i > x'_i$, then $$\interp{f}(x_1, \ldots, x_n) > \interp{f}(x_1, \ldots,x'_i,\ldots, x_n),$$
\item be weakly monotonic  if for any symbol $f$, the function $\interp{f}$ is a weakly monotonic function, that is if $x_i \geq x'_i$, then $$\interp{f}(x_1, \ldots, x_n) \geq \interp{f}(x_1, \ldots,x'_i,\ldots, x_n),$$
 \item have the weak sub-term property if for any symbol $f$, the function $\interp{f}$ verifies $\interp{f}(x_1, \ldots, x_n) \geq x_i$ with $i \in 1,\ldots,n$,\\
 \item to be strictly compatible (with the rewriting relation) if for all rules $\ell \to r$, $\interp{\ell} > \interp{r}$,
 \item to be weakly compatible if for all rules $\ell \to r$, $\interp{\ell} \geq \interp{r}$,
 \end{enumerate}
 \end{definition}
 
 \begin{definition}Given an ordered set $(A,<)$ and a program $\langle \Variables, \Constructors, \Functions, \Equations \rangle$, a  $(\Constructors \cup \Functions)$-algebra on $A$ is said to be a strict interpretation whenever it verifies (1), (3), (4). It is a quasi-interpretation whenever it verifies (2), (3), (5). It is a monotone interpretation whenever it verifies (2) and (5).
 \end{definition}
 
 Clearly, a strict interpretation is a quasi-interpretation which itself is a monotone interpretation.
  When we want to speak arbitrarily of one of those concepts, we use the generic word "interpretation".  We also use this terminology to speak about the function $\interp{f}$ given a symbol $f$.

Finally, by default, $A$ is chosen to be the set of real non negative numbers with its usual ordering. Moreover, we restrict the interpretations over the real numbers to be \emph{Max-Poly functions}, that is functions obtained by finite compositions of the constant functions, maximum, addition and multiplication. \MaxPoly\ denotes the set of these functions.

\begin{example}\label{booleanOp} Equality on binary words in $\{0,1\}^*$, boolean operations, membership in a list (built on $\cons, \nnill$) are computed as follows.
\begin{eqnarray*}
\epsilon = \epsilon & \to& \true\\
{\bf i}(x) = {\bf i}(y) & \to& x = y \text { with } {\bf i} \in \{0,1\}\\
{\bf i}(x) = {\bf j}(y) & \to& \false\text{ with }{\bf i} \neq {\bf j} \in \{0,1\}\\
\ffor(\true,y) &\to& \true\\
\ffor(\false,y) &\to& y\\
\fand(\true,y) &\to& y\\
\fand(\false,y) &\to& \false\\
 \iif \true \ \tthen y \ \eelse z & \to& y\\
 \iif \false\ \tthen y \ \eelse z &\to& z\\
\member(a,\nnill) &\to& \false\\
\member(a,\cons(b,l)) & \to&  \iif\! a=b \; \tthen\! \true \; \eelse \member(a,l)\\
\end{eqnarray*}

Such a program has the strict interpretation\footnote{To simplify the verification of inequalities, interpretations are taken in $[1,\infty[$.} given by:
\begin{eqnarray*}
\interp{\epsilon} = \interp{\true} = \interp{\false} = \interp{\nnill} &=& 1\\
\interp{{\bf i}}(x) &=& x+1\text{ with }{\bf i} \in \{0,1\}\\
\interp{\cons}(x,y) &=&x+y+1\\
\interp{=\!}(x,y) = \interp{\ffor}(x,y) = \interp{\fand}(x,y) &=& x+y\\
\interp{\iif \tthenÊ\eelse\!\!\!\!}(x,y,z) &=& x+y+z\\
\interp{\member}(x,y) &=&(x+2) \times y
\end{eqnarray*}
\end{example}

\begin{definition}The interpretation of a symbol $f$ is said to be additive if it has the shape $\sum_i x_i + c$. A program with an interpretation is said to be additive when its \emph{constructors} are additive.
\end{definition}

\subsection{Termination by Product Path Ordering}

Let us recall that the Product Path Ordering is a particular form of
the Recursive Path Orderings, a class of simplification orderings
(and so well-founded). Pioneers of this subject include
Plaisted~\cite{Plai78}, Dershowitz~\cite{Der82}, Kamin and
L\'evy~\cite{KL80}. 

Finally, let us mention that Krishnamoorthy and Narendran
in~\cite{KN85} have proved that deciding whether a program terminates
by Recursive Path Orderings is a NP-complete problem.

Let $\infEq_\Sigma$ be a preorder on a signature $\Sigma$, called \emph{quasi-precedence}
or simply \emph{precedence}. We write $\infSt_\Sigma$ for the induced strict precedence
and $\simeq_\Sigma$ for the induced equivalence relation on $\Sigma$. Usually, the 
context makes clear what $\Sigma$ is, and thus, we drop the subscript $\Sigma$.

\begin{definition}\label{def:Muliset-ordering}
  Given an ordering $\infEq$ over terms $\Terms{\Sigma}$, the product extension of $\infEq$ over sequences, written $\infMul$, 
  is defined as $(s_1, \ldots, s_k) \infMulSt (t_1, \ldots, t_k)$ iff 
  \begin{itemize}
  \item for all $i\leq k: s_i \infEq t_i$ and,
  \item  there is some $j\leq k$ such that $s_i \prec t_i$.
  \end{itemize}
where $\prec$ is the strict part of $\infEq$.  
\end{definition}

\begin{definition}\label{def:mpo} Given a program $\langle \Variables, \Constructors, \Functions, \Equations \rangle$ and 
  a precedence $\precEqFS$ over function symbols, the Product Path Ordering
 $\mpoSt$ is defined as the least ordering verifying rules given in Figure~\ref{fig:mpo}.
\end{definition} 

\begin{figure}
\hrule
\begin{gather*}
\ninfer{\termtwo = \termone_i \textit{\ or } \termtwo \mpoSt
  \termone_i}{\termtwo \mpoSt \funone(\ldots, \termone_i, \ldots)}{\funone \in
  \Functions \union \Constructors}
\\[5mm]
\ninfer{\forall i\ \termtwo_i \mpoSt
  \funone(\many{\termone}{n})}{\conone(\many{\termtwo}{m}) \mpoSt
  \funone(\many{\termone}{n})}{\funone \in \Functions, \conone \in
  \Constructors} \\[5mm]
\ninfer{\forall i\ \termtwo_i \mpoSt \funone(\many{\termone}{n})
  \qquad \funtwo \precFS \funone}{g(\many{\termtwo}{m}) \mpoSt
  \funone(\many{\termone}{n})}{\funone, \funtwo \in \Functions}
\\[5mm]
\ninfer{(\many{\termtwo}{n}) \mpoXSt
  (\many{\termone}{n}) \qquad \funone \egalFS \funtwo \qquad \forall
  i\ \termtwo_i \mpoSt
  \funone(\many{\termone}{n})}{\funtwo(\many{\termtwo}{n}) \mpoSt
  \funone(\many{\termone}{n})}{\funone, \funtwo \in \Functions}
\end{gather*}
\caption{Definition of $\mpoSt$} 
\label{fig:mpo}
\hrule
\end{figure}

\subsection{Characterizations in the confluent case}

\begin{theorem}[Bonfante, Cichon, Marion and Touzet~\cite{BonfanteCichonMarionTouzet01}]
Functions computed by programs with additive strict interpretation are exactly \Ptime \ functions.
\end{theorem}

It is Theorem 4 in~\cite{BonfanteCichonMarionTouzet01}, first item.

\begin{theorem}[Bonfante, Marion and Moyen~\cite{BMM07}]
Functions computed by programs with 
\begin{itemize}
\item an additive quasi-interpretation and 
\item a termination proof by \PPO
\end{itemize}
are exactly \Ptime \ functions.
\end{theorem}

The programs of this latter theorem are mentioned as $\text{RPO}^\text{QI}_{Pro}$-programs in~\cite{BMM07}, Theorem~48.

\section{Constructor preserving interpretations}

It is well known that the interpretations above can be used to bound both the length of the computations and the size of terms during the computations (see for instance~\cite{HL88,Shkaravska09}). Here we show that interpretations also cope with syntactic constraints. More precisely,  programs with polynomial constructor preserving interpretations generalize cons-free programs~\cite{Jones99}.

\newcommand{\stc}{\subseteq} 
\newcommand{\sts}{\sousterme^*} 

Let us consider a signature $\Constructors$, the signature of constructors in the sequel. $S(\Constructors)$ denotes the set of \emph{finite non-empty sets} of terms in $\Terms{\Constructors}$.
Let $\sousterme$ denotes the sub-term relation on terms. On $S(\Constructors)$, we define $\sts$ as follows: $m \sts m'$ iff $\forall t \in m : \exists t' \in m' : t \sousterme t'$.  As an ordering on sets, for interpretations, we will use the inclusion relation. One may observe that $m \stc m' \Longrightarrow m \sts m'$.  

\begin{definition}\label{def:cons}
Let us consider a monotone interpretation $\interp{-}$ of a program $\langle \Variables, \Constructors, \Functions, \Equations \rangle$ over $(S(\Constructors),\subseteq)$. We say that it preserves constructors if 
\begin{enumerate}
\item 
for any constructor symbol $\conone \in \Constructors$, 
$\interp{\conone}(m_1, \ldots,m_k) = \{ \conone(t_1, \ldots, t_k) \mid t_i \in m_i, i = 1, \ldots, k\}.$
\item given a rule $f(\many{\patone}{n})  \to r$ and a ground substitution $\sigma$, for all $u \sousterme r$, $$\interp{\sigma(u)}\sts \interp{\sigma(f(\many{\patone}{n}))}\cup_{i = 1}^n \interp{\sigma(\patone_i)}.$$ 
\end{enumerate}
By extension, we say that a program is constructor preserving if it admits a constructor preserving monotone interpretation. 
\end{definition}

The fact that a program preserves constructors fixes the definition of the interpretation over constructors. Moreover, (1) below gives a simple characterization of the interpretations of constructor terms.

\begin{proposition}\label{prop:pb:constructors}
\label{pr:cons}For any constructor preserving interpretation $\interp{-}$, the following holds:
\begin{enumerate}
\item for any ground constructor term $t$, $\interp{t} = \{ t\}$,
\item given a ground substitution $\sigma$, for any constructor terms $u \sousterme v$,  $\interp{\sigma(u)} \sts \interp{\sigma(v)}$.
\end{enumerate}
\end{proposition}

\begin{proof}(1) is proved by induction on the structure of terms. (2) is by induction on the structure of $v$. Suppose $v$ is a variable, if $u \sousterme v$, then $u = v$ and the property holds trivially. Suppose $v = \conone(v_1, \ldots, v_k)$. The case $u = v$ is as above. Otherwise, $u \sousterme v_j$ for some $j\leq k$. In that case, for all $t \in \interp{\sigma(u)}$, by induction, there is a $w_j \in \interp{\sigma(v_j)}$ such that $t \sousterme w_j$. Let us choose some $w_i \in \interp{\sigma(v_i)}$ for all $i \neq j$. Then, $t \sousterme \conone(w_1, \ldots, w_k) \in \interp{\sigma(\conone(v_1, \ldots, v_k))}$.
\end{proof}

\begin{proposition}\label{prop:nf:cp}
Given a program $\langle \Variables, \Constructors, \Functions, \Equations \rangle$ and a constructor preserving monotone interpretation $\interp{-}$, for all constructor terms $t_1, \ldots, t_n$ and all symbols $f$ of arity $n$, if $\sem{f}(t_1, \ldots, t_n) = t$, then $t \in \interp{f(t_1, \ldots, t_n)}$. 
\end{proposition}

\begin{proof}
For a monotone interpretation, if $u \to v$, then $\interp{v} \stc \interp{u}$. 
Suppose that $f(t_1, \ldots, t_n) \normto t$ with $t$ a constructor term, then, $\interp{t} \stc \interp{f(t_1, \ldots, t_n)}$. 
But, due to Proposition~\ref{prop:pb:constructors}-(1), $\interp{t} = \{t \}$. The conclusion follows.
\end{proof}

\begin{definition}In the present context, an interpretation is said to be polynomially bounded if for any symbol $f$, for any sets  $m_1, \ldots, m_n$,  the set $\interp{f}(m_1, \ldots, m_n)$ has a size polynomially bounded w.r.t. to the size of the $m_i$'s.  The size of a set $m$ is defined to be $|m| = \sum_{t \in m} |t|$.
\end{definition}

\begin{proposition}\label{prop:cons:cp}Given a constructor preserving program, then, for all constructors $\conone$, the size of the set $\interp{\conone}(m_1, \ldots, m_n)$ is polynomially bounded w.r.t. the size of the $m_i$'s.
\end{proposition}

\begin{proof} Let us write $M = \sum_{i = 1}^n |m_i|$. Then,
\begin{eqnarray*}
|\interp{\conone}(m_1, \ldots, m_n)| &\leq& \sum_{i \leq n, t_i \in m_i} |\conone(t_1, \ldots, t_n)|\\
&\leq& \sum_{i \leq n, t_i \in m_i} n \times M + 1 \qquad \text{ since } |t_i| \leq |m_i| \leq M\\
&\leq& M^n \times (n \times M+1)\qquad\text{ by a rough enumeration of }\\
&&\phantom{M^n \times (n \times M+1)}\qquad \text{ the indices of the sum.} 
\end{eqnarray*}
\end{proof}

\begin{example}Let us come back to Example~\ref{booleanOp}, it has a polynomially bounded constructor preserving monotone interpretation. Apart from the generic interpretation on constructors, we define:
\begin{eqnarray*}
\interp{=}(m,m')   = \interp{\member}(m,m') &=& \{ \true, \false\}\\
\interp{\fand}(m,m') = \interp{\ffor}(m,m') &=& m' \cup \{\true, \false\}\\
\interp{\iif \tthen \eelse\!\!\!\!}(m_b,m_y,m_z) &=& m_y \cup m_z
\end{eqnarray*}
It is clear that this interpretation is polynomially bounded. 
\end{example}

Actually, the notion of constructor preserving programs generalizes the notion of constructor-free programs as introduced by Jones (see for instance~\cite{Jones99}). He has shown how constructor-free programs characterize \Ptime \ and \Logspace. We recall  that a program is constructor-free whenever, for any rule $f(\many{\patone}{n}) \to r$, for any subterm $t \sousterme r$, 
\begin{itemize}
\item if $t$ is a constructor term, then $t \sousterme f(\many{\patone}{n})$, 
\item otherwise, the root of $t$ is not a constructor symbol.
\end{itemize}

\begin{proposition}
\label{pr:cons-pres}
Any \emph{constructor-free} program $\langle \Variables, \Constructors, \Functions, \Equations \rangle$  has  a polynomially bounded constructor preserving monotone interpretation $\interp{-}$. 
\end{proposition}

\begin{proof}We use the generic definition for constructor symbols. For functions, let $$\interp{f}(m_1,\ldots,m_n) = \{ u \mid \exists i \leq n, t \in m_i : u \sousterme t\}.$$

We have to prove a) that it is a monotone interpretation over $S(\Constructors)$, b) that it preserves constructors and c) that it is polynomially bounded. 

\paragraph{Proof of c).} Due to Proposition~\ref{prop:pb:constructors}, it is sufficient to verify the size condition on function symbols. Given some sets $m_1, \ldots, m_n$, we define $m = \{ u \mid \exists i \leq n, t \in m_i : u \sousterme t\} $ and $K =  |\cup_{i = 1}^n m_i|$. From the definition of the size of a set, for all $j \leq n$, $|m_j| \leq K \leq \sum_{i = 1}^n |m_i|$. Let $S_{m_1, \ldots, m_n} = \{ t \mid \exists i \leq n : t \in m_i\}$.  Then, 
\begin{eqnarray}
\# S_{m_1, \ldots, m_n} \leq \sum_{i = 1}^n \# m_i \leq n \times K\label{eq2}
\end{eqnarray}
where $\#m$ denotes the cardinality of a set $m$ (recall that for any set $m$ : $\#m \leq |m|$!). 
For each term $t$, let $D_t = \{ u \mid u \sousterme t\}$. Since  $m = \cup_{t \in S_{m_1,\ldots, m_n}} D_t$,  
$|m| \leq \sum_{t \in S_{m_1, \ldots, m_n}} |D_t|$.
It is clear that for all $t$, $\# D_t = |t|$. Moreover, each $u \sousterme t$ has a size smaller than $t$. Consequently, $|D_t| \leq |t|^2$. Since for all terms $t \in S_{m_1, \ldots, m_n}$, $|t| \leq K$, we have $|m| \leq \sum_{t \in S_{m_1, \ldots, m_n}} K^2$. Combining this latter equation with Equation~\ref{eq2}, we can state that $|m| \leq n \times K^3 \leq n \times (\sum_{i = 1}^n |m_i|)^3$.

\paragraph{Proof of b).} Let us come back to the  Definition~\ref{def:cons}. The first item comes from our generic choice for the interpretations of constructor symbols. Concerning the second item, let us consider a rule $f(\many{\patone}{n}) \to r$, a ground substitution $\sigma$ and  a subterm $u \sousterme r$.  We prove actually a stronger fact than condition 2, namely: $\interp{\sigma(u)} \stc \interp{\sigma(f(\many{\patone}{n}))}$. By induction on $u$.

 If $u$ is a variable or more generally a constructor term, since the program is constructor-free,  $u \sousterme \patone_j$ for some $j$. Take $t \in \interp{\sigma(u)}$. Due to Proposition~\ref{prop:pb:constructors}-(2),  there is $t' \in \interp{\sigma(\patone_j)}$ such that $t \sousterme t'$.  Then, recalling the definition of $\interp{f}$, $t$ belongs to $\interp{f}(\interp{\sigma(\patone_1)},\ldots, \interp{\sigma(\patone_n)}) = \interp{\sigma(f(\many{\patone}{n}))}$. 
 
 Otherwise, $u = g(v_1, \ldots, v_k)$ and, since the program is constructor-free, $g$ is a function symbol. Take $t \in \interp{\sigma(u)} = \interp{g}(\interp{\sigma(v_1)}, \ldots, \interp{\sigma(v_k)})$. Recall that $g$ is a function symbol. Then, by definition of $\interp{g}$, there is a $j\leq k$ and a term $t' \in \interp{\sigma(v_j)}$ such that  $t \sousterme t'$. By induction, $t' \in \interp{\sigma(f(\many{\patone}{n}))}$. But then, by definition of $\interp{f}$, $t \in \interp{f(\many{\patone}{n})}$.

\paragraph{Proof of a),} Item (2) of Definition~\ref{def:interpretation} is a direct consequence of the definition of the interpretation. Let us justify now (5).  As seen above,  for all rules $f(\many{\patone}{n}) \to r$, ground substitutions $\sigma$ and  subterms $u \sousterme r$, $\interp{\sigma(u)} \stc \interp{\sigma(f(\many{\patone}{n})}$. In particular, the result holds for $r$.
  \end{proof}

Do those kind of interpretations really go beyond constructor-freeness? Here is an example of a program which is not constructor-free, but with a constructor preserving interpretation.

\begin{example}Using the tally numbers $\zero, \suc$, and lists, the function $\tt f$ builds the list of the first $n-1$ integers given the argument $n$.
\begin{eqnarray*}
{\tt f}(\zero) &=&\nnill\\
{\tt f}(\suc(n)) &=&\cons(n,{\tt f}(n))
\end{eqnarray*}

Such a program has a constructor preserving interpretation. Let 
\begin{eqnarray*}
\interp{f}(m) &= &\{ \cons(n_1,\cons(n_2,\cdots(\cons(n_k,\nnill))\cdots)) \mid \\ 
&&\qquad \exists n_1, \ldots, n_k \in m : \forall j \leq k-1 : n_j  = \suc(n_{j+1}) \} \\
&&\cup \{\nnill\}.
\end{eqnarray*}

It is clear that this program is not constructor-free. But, there is a stronger difference:  the function computed by this program cannot be computed by \emph{any} constructor-free program. Indeed,
recall that the output of functions computed by constructor-free programs are subterms of the inputs. Since this is not the case of $\tt f$, the conclusion follows.

Let us make one last observation about the example. Actually, the interpretation is polynomially bounded. Indeed, a list of the shape $$\cons(n_1,\cons(n_2,\cdots(\cons(n_k,\nnill))\cdots))$$ is fixed by the choice of $n_1$ and $k$. Since $k \leq |n_1| \leq |m|$, there are at most $|m|^2$ such lists, each of which has a polynomial size.  
\end{example}

\begin{theorem}\label{th:cp}
Predicates computed by programs with a polynomially bounded constructor preserving interpretation are exactly \Ptime \ predicates.
\end{theorem}



\begin{proof}
From Jones's result and Proposition~\ref{pr:cons-pres}, it is clear that \Ptime \ predicates can be computed by constructor preserving programs. 

In the other direction, suppose we want to evaluate $f(t_1, \ldots, t_n)$ where $t_1, \ldots, t_n$ are some constructor terms. First, due to Propostion~\ref{prop:nf:cp}, one observes that the set of constructor terms $\interp{f(t_1, \ldots, t_n)} = \interp{f}(\interp{t_1}, \ldots, \interp{t_n})$ contains the normal form of $f(t_1, \ldots, t_n)$.  Moreover, this set has a polynomial size w.r.t. the size of $\interp{t_i}$'s. Due to Proposition~\ref{pr:cons}, $\interp{t_i} = \{t_i \}$ and consequently $|\interp{t_i}| = |t_i|$. That is $\interp{f(t_1, \ldots, t_n)}$ has a polynomial size w.r.t. the size of inputs.

As this is done by Jones, we use a call-by-value semantics with cache, that is:
\begin{itemize}
\item  we restrict substitutions to ground constructor substitutions,
\item each time a term $g(u_1, \ldots, u_m)$ is evaluated, it is put in a map $(g,u_1,\ldots, u_m) \mapsto \sem{g}(u_1, \ldots, u_m)$. This map is called the cache.
\end{itemize}
The key point to prove that computations can be done in polynomial time is to show that the cache has a polynomial size w.r.t. the size of the inputs. We begin to establish that for all constructor term $u_i$ such that $f(t_1, \ldots, t_n) \transto C[g(u_1, \ldots, u_m)]$: 
\begin{itemize}
\item there is a term $t \in \interp{f(t_1, \ldots, t_n)} \cup \{ t_1, \ldots, t_n\}$ such that $u_i \sousterme t$,
\item $\interp{g(u_1, \ldots, u_m)} \sts \interp{f(t_1, \ldots, t_n)} \cup \{ t_1, \ldots, t_n\}$.
\end{itemize}

One will have noticed that $\{ t_1, \ldots, t_n\} = \cup_{i = 1}^n \interp{t_i}$, so that $ \interp{f(t_1, \ldots, t_n)} \cup \{ t_1, \ldots, t_n\} = \interp{f(t_1, \ldots, t_n)} \cup_{i = 1}^n  \interp{t_i}$. Second remark, terms like $f(t_1, \ldots, t_n)$, $g(u_1, \ldots, u_m)$ as above correspond to nodes in the call tree with $(f,t_1, \ldots, t_n) \too^+ (g, u_1, \ldots, u_m)$.
So, we work by induction on $\too^+$. 

\paragraph{Base case.}
Suppose that $(f,t_1, \ldots, t_n) \too (g, u_1, \ldots, u_m)$. In other words, there is a context $C$ such that:
$$f(t_1, \ldots, t_n) \to  C[g(v_1, \ldots, v_m)] \reftransto C[g(u_1, \ldots, u_m)].$$  By Lemma~\ref{lem:sousterme} below, $u_i \sousterme t$ for some term $t \in  \interp{f(t_1, \ldots, t_n)} \cup \{ t_1, \ldots, t_n\}$ as required. 

For the second item, notice that $g(v_1, \ldots, v_m) \sousterme C[g(v_1, \ldots, v_m)]$. Then, 
$$\begin{array}{rll}
\interp{g(u_1, \ldots u_m)} &\stc \interp{g(v_1, \ldots, v_m)}&\text{ since }g(v_1, \ldots, v_m) \reftransto g(u_1, \ldots, u_m)\\
&\sts \interp{f(t_1, \ldots, t_n)} \cup \{ t_1, \ldots, t_n\}&\text{ by Definition~\ref{def:cons}, second item}
\end{array}$$

\paragraph{Induction step.}
 Otherwise, $(f,t_1, \ldots, t_n) \too^+ (g,u_1, \ldots, u_m) \too (h,w_1, \ldots, w_k)$. By induction, we have $\interp{g(u_1, \ldots, u_m)} \sts \interp{f(t_1, \ldots, t_n)} \cup \{ t_1, \ldots, t_n\}$ and for all $i \leq m$, $\{ u_i \} = \interp{u_i} \sts \interp{f(t_1, \ldots, t_n)} \cup \{ t_1, \ldots, t_n\}$.  Consequently, 
 \begin{equation}
 \interp{g(u_1, \ldots, u_m)} \cup_{i=1}^m \interp{u_i} \sts \interp{f(t_1, \ldots, t_n)} \cup \{ t_1, \ldots, t_n\}.\label{eq23}
 \end{equation}
  
 By Lemma~\ref{lem:sousterme}, for all $w_i$, there is a term $v \in \interp{g(u_1, \ldots, u_m)} \cup_{i = 1}^m \interp{u_i}$. By Equation~\ref{eq23}, there is a term $t \in  \interp{f(t_1, \ldots, t_n)} \cup \{ t_1, \ldots, t_n\}$ such that $w_i \sousterme t$. 
 
 For the second item, $$h(w_1, \ldots, w_k) \sts \interp{g(u_1, \ldots, u_m)} \cup_{i = 1}^m \interp{u_i} \sts \interp{f(t_1, \ldots, t_n)} \cup \{ t_1, \ldots, t_n\}.$$
 where the first relation is due to Definition~\ref{def:cons}-(2).

After this preliminary work, we are ready to bound the size of the cache. As a consequence of what precedes, the arguments of all the calls $g(u_1, \ldots, u_m)$ in the call tree are contained in the set $$S = \{u  \mid \exists t \in \interp{f(t_1, \ldots, t_n)} \cup \{ t_1, \ldots, t_n\} : u \sousterme t\}.$$ 
Since $\interp{f(t_1, \ldots, t_n)} \cup \{t_1,\ldots, t_n\}$ has a polynomial size, $S$ itself has cardinality bounded by a polynomial,  say $P(|t_1|, \ldots, |t_n|)$. As a consequence, since the $u_i$'s are in $S$,  the cache has at most $|\Functions| \times P(|t_1|, \ldots, |t_n|)^D$ entries, where  $D$ is a bound on the arity of symbols. Since each elements $u_i$ and each normal form of $g(u_1, \ldots, u_m)$ are subterms of $\interp{f(t_1, \ldots, t_n)} \cup \{t_1,\ldots, t_n\}$, they have a polynomial size. Then, the cache itself has a polynomial size.
 \end{proof}

\begin{lemma}\label{lem:sousterme} Let $\langle \Variables, \Constructors, \Functions, \Equations \rangle$ be a program with a polynomially bounded constructor preserving interpretation $\interp{-}$. Suppose given a rewriting step $f(t_1, \ldots, t_n) \to w$ with $t_1, \ldots, t_n$ some constructor terms and $v \sousterme w$. If $v \normto u$ with $u$ a constructor term, then there is a term $t \in \interp{f(t_1, \ldots, t_n)} \cup_{i = 1}^n \interp{t_i}$ such that $u\sousterme t$. 
\end{lemma}

\begin{proof}
Let $f(\many{\patone}{n}) \to r$ and $\sigma$ be such that $f(t_1, \ldots, t_n)  = \sigma(f(\many{\patone}{n})) \to  \sigma(r) = w$. 
There are two cases: if $v \sousterme \sigma(x)$ for some variable $x \in \patone_j$. Since the $t_i$ are constructor terms, $v$ is necessarily a constructor term, and consequently a normal form. So, $v = u$. As a matter of fact, $v \sousterme \sigma(\patone_j) = t_j$. We conclude taking $t = t_i$.  

Otherwise, $v = \sigma(v')$ for some $v' \sousterme r$. Since $v \normto u$, we have $\interp{u} \stc \interp{v}$. By Proposition~\ref{pr:cons}-(1), $u \in \interp{v}$. Due to Definition~\ref{def:cons}, second item, there is $t \in  \interp{f(t_1, \ldots, t_n)} \cup_{i = 1}^n \interp{t_i}$ such that $u\sousterme t$. 
\end{proof}

\section{Observation by non confluence}
\label{sec:4}

\subsection{Semantics}

We first have to define what we mean when we say that a function is
computed by a non-confluent rewriting system. Computations lead to 
several normal forms, depending on the reductions applied.
At first sight, we shall regard a non-confluent rewrite system as a non-deterministic 
algorithm. 

In this section, given a program, we do not suppose its underlying rewriting system to be confluent. By extension, we say that such programs are not confluent (even if they may be so). 

\newcommand{\bd}{\mbox{$\mathbb{D}$}}

\begin{example}
\label{sat}
A 3-SAT formula is given by a set of clauses, written $\vee(x_1,x_2,x_3)$ where the $x_i$ have either the shape $\neg(n_i)$ or $\ee(n_i)$.\footnote{ $\ee$ corresponds to a positive occurrence of a variable. It is introduced for a question of uniformity.} The $n_i$'s which are the identifiers of the variables are written in binary, with unary constructors $0, 1$ and the constant $\epsilon$. To simplify the program, we suppose all identifiers to have the same length. $\vrai, \faux$ represent the boolean values {\em true}
and {\em false}. $\vee$ serves for the disjunction. Since we focus on 3-SAT formulae, we take it to be a ternary function. 
For instance the formula $(x_1 \vee x_2 \vee \overline{x_3}) \wedge (x_1 \vee \overline{x_2} \vee \overline{x_1})$ is represented as:
$$\begin{array}{l}
\cons(\vee(\ee(0(1(\epsilon))),\ee(1(0(\epsilon))),\neg(1(1(\epsilon)))),\\
\phantom{\cons(}\cons(\vee(\ee(0(1(\epsilon))),\neg(1(0(\epsilon))),\neg(0(1(\epsilon)))),\nnill)).
\end{array}$$

\vspace{2ex}
Recalling rules given in Example~\ref{booleanOp}, the following program computes the satisfiability of a formula. Let us suppose  that $\ell$ denotes the list of variables with the valuation "true", we have:

$$
\begin{array}{rlp{2mm}rl}
 \verify(\nnill,\ell) &\to \true\\
 \verify(\cons(\vee(x_1,x_2,x_3),\psi),\ell) &\to \fand(\ffor(\ffor(\eval(x_1,\ell),\eval(x_2,\ell)),\eval(x_3,\ell)),\verify(\psi,\ell))\\
 \eval(\neg(n),\ell) &\to \iif \member(n,\ell) \ \tthen \false \ \eelse \true\\
 \eval(\ee(n),\ell) &\to \iif \member(n,\ell) \ \tthen \true \ \eelse \false\\
%
%
\end{array}
$$

It is sufficient to compute the set of "true" variable. This is done by the rules:
\begin{eqnarray*}
\hypo(\nnill) &=&\nnill\\
\hypo(\cons(\vee({\bf a}(x_1),{\bf b}(x_2),{\bf c}(x_3)),\ell)) &\to& \hypo(\ell)\\
\hypo(\cons(\vee({\bf a}(x_1),{\bf b}(x_2),{\bf c}(x_3)),\ell)) &\to& \cons(x_i,\hypo(\ell))\\
\hypo(\cons(\vee({\bf a}(x_1),{\bf b}(x_2),{\bf c}(x_3)),\ell)) &\to& \cons(x_i,\cons(x_j,\hypo(\ell)))\\
\hypo(\cons(\vee({\bf a}(x_1),{\bf b}(x_2),{\bf c}(x_3)),\ell)) &\to& \cons(x_i,\cons(x_j,\cons(x_k,\hypo(\ell))))\\
\main(\psi) &\to&\verify(\psi,\hypo(\psi))
\end{eqnarray*}
with ${\bf a}, {\bf b}, {\bf c} \in \{\neg, \ee\}$ and  $i \neq j \neq k \in \{1,2,3\}$. The main function is $\main$.

The rules involving $\hypo$ are not confluent, and correspond exactly to
the non-deterministic choice. (By Newman's Lemma, the systems
considered are not weakly confluent since they are terminating.)

Such a program has an interpretation, given by:
\begin{eqnarray*}
\interp{\neg}(x) = \interp{\ee}(x) &=&x+1\\
\interp{\vee}(x_1,x_2,x_3) &=& x_1+x_2+x_3+10\\
\interp{\eval}(x,y) &=&(x+1) \times y + 3\\
\interp{\verify}(x,y) &=& (x+1) \times (y+1)\\
\interp{\hypo}(x) &= &x+1\\
\interp{\main}(x) &=& \interp{\verify(x,\hypo(x))}+1
\end{eqnarray*}
\end{example}

Our notion of computation 
by a non-confluent system 
appears in Krentel's work~\cite{Krentel}, in a different context. It seems
appropriate and robust, as argued by Gr\"adel and Gurevich~\cite{GurevichGradel}. 

We suppose given a linear order $\prec$ on symbols, this order can be extended
to terms using the lexicographic ordering. We use the same notation $\prec$ for this order. Then, we 
say that a (partial) function $\phi : \Terms{\Constructors}^m \to \Terms{\Constructors}$
 is computed by a program $\langle \Variables, \Constructors, \Functions, \trs\rangle$ if for all $t_1, \ldots, t_m \in \Terms{\Constructors}$:
$$\phi(t_1, \ldots, t_n) \text{ is defined} \equivalent \phi(t_1, \ldots, t_n) = \max_\prec \{ v \mid f(t_1, \ldots, t_n) \normto v \}$$ 


In some case, we get the expected result: non-confluence corresponds exactly to
 non-determinism. Confluent programs with an interpretation
compute \Ptime, and the non-confluent ones compute $\NPtime$.


\begin{theorem}[Bonfante, Cichon, Marion and Touzet~\cite{BonfanteCichonMarionTouzet01}]\label{th:nd:pi}
Functions computed by non confluent programs with an additive polynomial interpretation are exactly 
$\NPtime$ functions;
\end{theorem}


\subsection{Non confluent programs with a polynomial quasi-interpretation}

 



The following result is more surprising.

 \begin{theorem}\label{th:mpo:n}Functions computed by non confluent programs that admit a quasi-interpretation and a \PPO \ proof of termination are exactly \Pspace \ functions.
\end{theorem}

The proof of the theorem essentially relies on the following example:

\newcommand{\Bool}{\mathbb{B}}

\begin{example}\label{qbf:mpo:n}[Quantified Boolean Formula] Let us compute the problem of the Quantified Boolean Formula.
 The principle of the algorithm is in two steps, the first one is top-down, the second one is bottom-up. 
In the first part, we span the computation to the leaves where we make an hypothesis on the value
of (some of) the variables. In the second part, coming back at the top, we compute the truth value of the formula and verify that the hypotheses chosen in the different branches are compatible between them. 

As above, we suppose that variables are represented by binary strings on constructors $0,1,\epsilon$. To them, we add the constructors $\cons, \vide$ to build lists,  $\blanc, \brouge, \noir$ to decorate variables. Given a variable $n$, the decorations give the truth value of the variables. $\brouge(n)$ corresponds to an unchosen value, that is true or false, 
$\noir(n)$ corresponds to false, and $\blanc(n)$ to true. The booleans are $\true, \false$ and $\bot$ serves for trash.  $\utrue$ and $\ufalse$ are two (unary) constructors representing booleans within computations. To help the reader, we give an informal type to the key functions:  $\Psi$ corresponds to formulae, $\Lambda$ to lists (of decorated variables) and $B$ to truth values.   Truth values are terms of the shape $\utrue(\Lambda)$ or $\ufalse(\Lambda)$. Finally, $V$ is the type of variables and $\Bool = \{ \true,\false\}$.
The following rules correspond to the first step of the computation. $\main: \Psi \to B$, $\verify: \Psi \times \Lambda \to B$, $\fnot: B\to B$, $\ffor: B\times B \to B$, $\hyp:B \times V \times \Bool \to B$, $\Put: \Lambda \times V \times \Bool \to \Lambda$  and $\listHyp: \Psi \to \Lambda$:

\begin{eqnarray*}
\main(\phi) & \to & \verify(\phi, \listHyp(\phi))\\
\verify(\cvar(x),h) & \to & \utrue(\Put(h,x,\true))\\
\verify(\cvar(x),h) & \to & \ufalse(\Put(h,x,\false))\\
\verify(\Cor(\phi_1,\phi_2),h) & \to & \ffor(\verify(\phi_1,h),
\verify(\phi_2,h))\\
\verify(\cnot(\phi),h) & \to & \fnot(\verify(\phi,h))\\
\verify(\cexists(x,\phi),h) & \to & \ffor(\hyp(\verify(\phi,h),x,\true),
\hyp(\verify(\phi,h),x,\false))\\
\end{eqnarray*}
where $h$ is a valuation of the variables.
The rules $\verify(\cvar(x),h) \to \utrue(\Put(h,x,\true))$ and 
$\verify(\cvar(x),h) \to  \ufalse(\Put(h,x,\false))$
are the unique rules responsible of the non confluence of the program. This is the step
where the value of variables is actually chosen. Concerning the valuations, they are 
written as lists $\cons(\text{tv}(n), \cons(\cdots))$ where the truth value $\text{tv}$ of a variable 
is in $\{ \brouge, \blanc, \noir\}$. 

Suppose that $x$ is a variable ocuring in $\Cor(\phi_1, \phi_2)$. 
One key feature is that in a computation of $\verify(\Cor(\phi_1,\phi_2),h) \to \ffor(\verify(\phi_1,h),
\verify(\phi_2,h))$, the choice of the truth value of the variable $x$ can be different in the
two sub-computation $\verify(\phi_1, h)$ and $\verify(\phi_2,h)$. 
Then, the role of the bottom-up part of the computation is to verify
that these choices are actually compatible.

The two functions $\Put$ and $\listHyp$ are computed by:
\begin{eqnarray*}
\listHyp(\cvar(x)) & \to &\vide\\
\listHyp(\Cor(\phi_1,\phi_2)) & \to&  \append(\listHyp(\phi_1),
\listHyp(\phi_2))\\
\listHyp(\cnot(\phi)) & \to&  \listHyp(\phi)\\
\listHyp(\cexists(x,\phi)) & \to & \cons(\brouge(x),\listHyp(\phi))\\
\Put(\cons(\brouge(n),l),m,\true) &\To& \iif n = m\ \tthen \cons(\blanc(n),l)\\
&&\eelse \cons(\brouge(n),\Put(l,m,\true))\\
\Put(\cons(\brouge(n),l),m,\false) &\To& \iif n = m\ \tthen \cons(\noir(n),l)\\
&&\eelse \cons(\brouge(n),\Put(l,m,\false))\\
\end{eqnarray*}

Then, the computation returns back. The top-down part of the computation returned a ``tree'' whose interior 
nodes are labeled with ``$\ffor$''
and ``$\fnot$''. At the leaves, we have $\utrue(l)$ or $\ufalse(l)$ where $l$ stores the
truth value of variables.  The logical rules are:

$$\begin{array}{lcrp{5mm}lcr}
\fnot(\ufalse(x)) & \to& \utrue(x) &&\ffor(\utrue(x),\utrue(y))& \to&   \utrue(\Union(x,y))\\
\fnot(\utrue(x)) &\to& \ufalse(x) &&\ffor(\ufalse(x),\utrue(y)) &\To& \utrue(\Union(x,y))\\
&&&&\ffor(\utrue(x),\ufalse(y))& \to  & \utrue(\Union(x,y))\\
&&&&\ffor(\ufalse(x),\ufalse(y)) &\To& \ufalse(\Union(x,y))
\end{array}$$

\noindent where $x$  and $y$ correspond to the list of hypothesis. The function
$\Union$ takes two lists and verify that they made compatible  hypotheses 
on the truth value of variables. $\brouge(n)$ is compatible with
both $\noir(n)$ and $\blanc(n)$. But $\noir(n)$ and $\blanc(n)$ are not compatible.
 
\begin{eqnarray*}
\Union(\vide,\vide) &\To& \vide\\
\Union(\cons(x,l),\cons(x,l')) &\To& \cons(x,\Union(l,l'))\\
\Union(\cons(\blanc(x),l),\cons(\brouge(x),l')) &\To& \cons(\blanc(x),\Union(l,l'))\\
\Union(\cons(\noir(x),l),\cons(\brouge(x),l')) &\To& \cons(\noir(x),\Union(l,l'))\\
\Union(\cons(\brouge(x),l),\cons(\blanc(x),l')) &\To& \cons(\blanc(x),\Union(l,l'))\\
\Union(\cons(\brouge(x),l),\cons(\noir(x),l')) &\To& \cons(\noir(x),\Union(l,l'))\\
\Union(\cons(\blanc(x),l),\cons(\noir(x),l')) &\To& \bot\\
\Union(\cons(\noir(x),l),\cons(\blanc(x),l')) &\To& \bot\\
\end{eqnarray*}

The matching process runs only for lists of equal length and variables must be presented in the same order. 
This hypothesis is fulfilled for our program. 
The last verification corresponds to the $\cexists$ constructor. For the left branch of the $\ffor$, 
the variable $x$ is supposed to be true, for the second branch it is supposed to be false. This 
verification is performed by the $\hyp$ function.

\begin{eqnarray*}
\hyp(\utrue(h),x,y) & \To &\utrue(\hyp(h,x,y))\\
\hyp(\ufalse(h),x,y) & \To &\ufalse(\hyp(h,x,y))\\
\hyp(\cons(\blanc(y),l),x,\true) & \To & \iif x = y\ \tthen l\ \\&&\eelse \cons(\blanc(y),\hyp(l,x,\true))\\
\hyp(\cons(\noir(y),l),x,\true) & \To &\iif x = y\ \tthen \undef\\&& \eelse \cons(\noir(y),\hyp(l,x,\true))\\
\hyp(\cons(\noir(y),l),x,\false) & \To &\iif x = y\ \tthen l\\&&  \eelse \cons(\noir(y),\hyp(l,x,\false))\\
\hyp(\cons(\blanc(y),l),x,\false) & \To &\iif x = y\ \tthen \undef\\&& \eelse \cons(\blanc(y),\hyp(l,x,\false))
\end{eqnarray*}

The rules for $\iif \tthen \eelse\!\!$, for $=$ and for $\append$ are omitted.

It is then routine to verify that this program is ordered by \PPO. The order $\main \succ \verify \succ \hyp \succ \Put \succ \ffor \succ \fnot \succ \listHyp \succ \Union \succ \iif \succ \append \succ =$ is compatible with the rules. 

To end the Example, we provide a quasi-intepretation. $\interp{\verify}(X,H) = X+H$, $\interp{\main}(\Phi) = 2\Phi$, and for all other function symbols we take $\interp{f}(X_1, \ldots, X_n) = \max(X_1, \ldots, X_n)$. For constructors, we take for all of them $\interp{\conone}(X_1, \ldots, X_n) = \sum_{i = 1}^n X_i + 1$.
\end{example}

\begin{proof}[of Theorem~\ref{th:mpo:n}]
Let us begin with the following proposition.

\begin{proposition}$\functional.\text{QI}.\PPO$ is closed by composition. That is if $f, g_1, \ldots g_k$ are computable with programs in $\functional.\text{QI}.\PPO$, then, the function $$\lambda x_1, \ldots, x_n.f(g_1(x_1, \ldots, x_n), \ldots, g_k(x_1, \ldots, x_n))$$ is computable by a program in   $\functional.\text{QI}.\PPO$.
\end{proposition}

\begin{proof}
One adds a new rule $h(x_1, \ldots, x_n) \to f(g_1(x_1, \ldots, x_n), \ldots, g_k(x_1, \ldots, x_n))$ with precedence $h \succ f$ and $h \succ g_i$ for all $i\leq k$. The rule is compatible with the interpretation: $$\interp{h}(x_1, \ldots, x_n) = \interp{f(g_1(x_1, \ldots, x_n), \ldots, g_k(x_1, \ldots, x_n))}.$$
\end{proof}

The Example~\ref{qbf:mpo:n} shows that a \Pspace-complete problem can be solved in the considered class of programs.  By composition of QBF with the reduction, since polynomial time functions can be computed in $\functional.\text{QI}.\PPO.n$,  any \Pspace \ predicate can be computed in $\functional.\text{QI}.\PPO.n$. Let us recall now that computing bit $i$th of the output of  a \Pspace \ function is itself computable in \Pspace. Since, building the list of the first integers below some polynomials can be computed in polynomial time, by composition, the conclusion follows.
\end{proof}

\subsection{Non confluent constructor preserving programs}

\begin{theorem}Predicates computed by non confluent programs with a polynomially bounded constructor preserving interpretation are exactly the $\Ptime$-computable predicates.
\end{theorem}

This result is close to the one of Cook in~\cite{Cook71} (Theorem 2).  He gives a characterization of \Ptime \ by means of auxiliary pushdown automata working in logspace, that is a Turing Machine working in logspace plus an extra (unbounded) stack.   It is also the case that the 
result holds  whether or not the auxiliary pushdown automata is deterministic. 

The proof follows the line of~\cite{AMAST06}, we propose thus just a sketch of the proof. 
The key observation is that arguments of recursive calls are sub-terms of the initial interpretation, a property that holds for confluent programs. As a consequence, following a call-by-value semantics, any arguments in sub-computations are some sub-terms of the initial interpretations. From that, it is possible to use memoization, see~\cite{Jones97}. The original point is that we have to manage non-determinism.




The proof of Proposition~\ref{prop:nf:cp} holds for non confluent computations. So, normal forms of a term $t$ are in $\interp{t}$. As we have seen in the proof of Theorem~\ref{th:cp}, this set has a polynomial size wrt the size of the inputs. In the non deterministic case, the cache is still a map, with the same keys, but the values of the map enumerate the list of normal forms. With the preceding observation, we can state that the map has still a polynomial size.

\paragraph{Acknowledgement.} I'd like to thank the anonymous referees for their precious help. Their sharp reading has been largely valuable to rewrite some part of the draft.

\bibliographystyle{eptcs} 
\bibliography{bib}

\begin{thebibliography}{10}
\providecommand{\bibitemstart}[1]{\bibitem{#1}}
\providecommand{\bibitemend}{}
\providecommand{\bibliographystart}{}
\providecommand{\bibliographyend}{}
\providecommand{\url}[1]{\texttt{#1}}
\providecommand{\urlprefix}{Available at }
\providecommand{\bibinfo}[2]{#2}
\bibliographystart

\bibitemstart{AMAST06}
\bibinfo{author}{{G}uillaume {B}onfante} (\bibinfo{year}{2006}):
  \emph{\bibinfo{title}{{S}ome programming languages for {LOGSPACE} and
  {PTIME}}}.
\newblock In: {\sl \bibinfo{booktitle}{11th {I}nternational {C}onference on
  {A}lgebraic {M}ethodology and {S}oftware {T}echnology - {AMAST}'06}},
  \bibinfo{address}{{K}uresaare/{E}stonie}.
\bibitemend

\bibitemstart{BonfanteCichonMarionTouzet01}
\bibinfo{author}{Guillaume Bonfante}, \bibinfo{author}{Adam Cichon},
  \bibinfo{author}{Jean-Yves Marion} \& \bibinfo{author}{H\'el\`ene Touzet}
  (\bibinfo{year}{2001}): \emph{\bibinfo{title}{Algorithms with polynomial
  interpretation termination proof}}.
\newblock {\sl \bibinfo{journal}{J. Funct. Program.}}
  \bibinfo{volume}{11}(\bibinfo{number}{1}), pp. \bibinfo{pages}{33--53}.
\bibitemend

\bibitemstart{BMM07}
\bibinfo{author}{Guillaume Bonfante}, \bibinfo{author}{Jean-Yves Marion} \&
  \bibinfo{author}{Jean-Yves Moyen} (\bibinfo{year}{2009}):
  \emph{\bibinfo{title}{Quasi-interpretations: a way to control resources}}.
\newblock {\sl \bibinfo{journal}{Theoretical Computer Science}}
  \bibinfo{note}{To appear}.
\bibitemend

\bibitemstart{Cook71}
\bibinfo{author}{Stephen Cook} (\bibinfo{year}{1971}):
  \emph{\bibinfo{title}{Characterizations of pushdown machines in terms of
  time-bounded computers}}.
\newblock {\sl \bibinfo{journal}{Journal of the ACM}}
  \bibinfo{volume}{18}(\bibinfo{number}{1}), pp. \bibinfo{pages}{4--18}.
\bibitemend

\bibitemstart{Der82}
\bibinfo{author}{Nachum Dershowitz} (\bibinfo{year}{1982}):
  \emph{\bibinfo{title}{Orderings for term-rewriting systems}}.
\newblock {\sl \bibinfo{journal}{Theoretical Computer Science}}
  \bibinfo{volume}{17}(\bibinfo{number}{3}), pp. \bibinfo{pages}{279--301}.
\bibitemend

\bibitemstart{DJ90}
\bibinfo{author}{Nachum Dershowitz} \& \bibinfo{author}{Jean-Pierre Jouannaud}
  (\bibinfo{year}{1990}): \emph{\bibinfo{title}{Handbook of Theoretical
  Computer Science vol.B}}, chapter \bibinfo{chapter}{Rewrite systems}, pp.
  \bibinfo{pages}{243--320}.
\bibitemend

\bibitemstart{Gaboardi08a}
\bibinfo{author}{Marco Gaboardi}, \bibinfo{author}{Jean-Yves Marion} \&
  \bibinfo{author}{Simona Ronchi Della~Rocca} (\bibinfo{year}{2008}):
  \emph{\bibinfo{title}{A logical account of pspace}}.
\newblock {\sl \bibinfo{journal}{SIGPLAN Not.}}
  \bibinfo{volume}{43}(\bibinfo{number}{1}), pp. \bibinfo{pages}{121--131}.
\bibitemend

\bibitemstart{GurevichGradel}
\bibinfo{author}{Erich Gr{\"a}del} \& \bibinfo{author}{Yuri Gurevich}
  (\bibinfo{year}{1995}): \emph{\bibinfo{title}{Tailoring recursion for
  complexity}}.
\newblock {\sl \bibinfo{journal}{Journal of symbolic logic}}
  \bibinfo{volume}{60}(\bibinfo{number}{3}), pp. \bibinfo{pages}{952--69}.
\bibitemend

\bibitemstart{HL88}
\bibinfo{author}{Dieter Hofbauer} \& \bibinfo{author}{Clemens Lautemann}
  (\bibinfo{year}{1988}): \emph{\bibinfo{title}{Termination proofs and the
  length of derivations}}.
\newblock {\sl \bibinfo{journal}{Lecture Notes in Computer Science}}
  \bibinfo{volume}{355}, pp. \bibinfo{pages}{167--177}.
\bibitemend

\bibitemstart{Jones97}
\bibinfo{author}{Neil Jones} (\bibinfo{year}{1997}):
  \emph{\bibinfo{title}{Computability and complexity, from a programming
  perspective}}.
\newblock \bibinfo{publisher}{MIT Press}.
\bibitemend

\bibitemstart{Jones99}
\bibinfo{author}{Neil Jones} (\bibinfo{year}{1999}):
  \emph{\bibinfo{title}{LOGSPACE and PTIME characterized by programming
  languages}}.
\newblock {\sl \bibinfo{journal}{Theroretical Computer Science}}
  \bibinfo{volume}{228}, pp. \bibinfo{pages}{151--174}.
\bibitemend

\bibitemstart{KL80}
\bibinfo{author}{Samuel Kamin} \& \bibinfo{author}{Jean-Jacques L\'evy}
  (\bibinfo{year}{1980}): \emph{\bibinfo{title}{Attempts for generalising the
  recursive path orderings.}}
\newblock \bibinfo{type}{Technical Report}, \bibinfo{institution}{Univerity of
  Illinois, Urbana}.
\newblock \bibinfo{note}{Unpublished note. {A}ccessible on
  http://perso.ens-lyon.fr/pierre.lescanne/not\_accessible.html}.
\bibitemend

\bibitemstart{Krentel}
\bibinfo{author}{Mark~W. Krentel} (\bibinfo{year}{1988}):
  \emph{\bibinfo{title}{The complexity of optimization problems}}.
\newblock {\sl \bibinfo{journal}{Journal of computer and system sciences}}
  \bibinfo{volume}{36}, pp. \bibinfo{pages}{490--519}.
\bibitemend

\bibitemstart{KN85}
\bibinfo{author}{Mukkai~S. Krishnamoorthy} \& \bibinfo{author}{Paliath
  Narendran} (\bibinfo{year}{1985}): \emph{\bibinfo{title}{On recursive path
  ordering}}.
\newblock {\sl \bibinfo{journal}{Theoretical Computer Science}}
  \bibinfo{volume}{40}(\bibinfo{number}{2-3}), pp. \bibinfo{pages}{323--328}.
\bibitemend

\bibitemstart{Kristiansen06}
\bibinfo{author}{Lars Kristiansen} (\bibinfo{year}{2006}):
  \emph{\bibinfo{title}{Complexity-Theoretic Hierarchies}}.
\newblock In: {\sl \bibinfo{booktitle}{CiE}}, {\sl \bibinfo{series}{Lecture
  Notes in Computer Science}} \bibinfo{volume}{3988},
  \bibinfo{publisher}{Springer}, pp. \bibinfo{pages}{279--288}.
\bibitemend

\bibitemstart{Kristiansen09}
\bibinfo{author}{Lars {K}ristiansen} \& \bibinfo{author}{Bedeho Mender}
  (\bibinfo{year}{2009}): \emph{\bibinfo{title}{The Semantics and Complexity of
  Successor-free non deterministic {G\"o}del {T} and {PCF}}}.
\newblock In: {\sl \bibinfo{booktitle}{Computability in Europe, CIE '09}},
  \bibinfo{address}{Heidelberg, Germany}.
\bibitemend

\bibitemstart{Plai78}
\bibinfo{author}{D.~Plaisted} (\bibinfo{year}{1978}): \emph{\bibinfo{title}{A
  recursively defined ordering for proving termination of term rewriting
  systems}}.
\newblock \bibinfo{type}{Technical Report} \bibinfo{number}{R-78-943},
  \bibinfo{institution}{Department of Computer Science, University of
  Illinois}.
\bibitemend

\bibitemstart{Shkaravska09}
\bibinfo{author}{Olha Shkaravska}, \bibinfo{author}{Marko van Eekelen} \&
  \bibinfo{author}{Ron van Kesteren} (\bibinfo{year}{2009}):
  \emph{\bibinfo{title}{Polynomial Size Analysis of First-Order Shapely
  Functions}}.
\newblock {\sl \bibinfo{journal}{CoRR}} \bibinfo{volume}{abs/0902.2073}.
\newblock \urlprefix\url{http://arxiv.org/abs/0902.2073}.
\bibitemend

\bibitemstart{terese}
\bibinfo{author}{Te{R}e{S}e} (\bibinfo{year}{2003}): \emph{\bibinfo{title}{Term
  Rewriting Systems}}, {\sl \bibinfo{series}{Cambridge Tracks in Theoretical
  Computer Science}}~\bibinfo{volume}{55}.
\newblock \bibinfo{publisher}{Cambridge University Press}.
\bibitemend

\bibliographyend
\end{thebibliography}

\end{document}